\documentclass{article}
\usepackage{fullpage}
\usepackage{epsfig}
\usepackage{graphics}
\usepackage{latexsym}
\usepackage{amsmath}
\usepackage{amsfonts}
\usepackage{amssymb}
\usepackage[mathscr]{eucal}
\usepackage{mathrsfs}
\usepackage{pifont}
\usepackage{yhmath}
\usepackage{undertilde}
\usepackage[usenames]{color}
\usepackage{amsthm}

\usepackage{url}

\urldef{\mailsa}\path|{xchen,xdhu,wcj}@amss.ac.cn|

\newtheorem{theorem}{Theorem}[section]
\newtheorem{lemma}[theorem]{Lemma}

\newtheorem{definition}[theorem]{Definition}
\newtheorem{algorithm}{Algorithm}{\itshape}{\rmfamily}

\newtheorem{example}[theorem]{Example}

 \theoremstyle{remark}

{\itshape}{\rmfamily}

{\itshape}{\rmfamily}

{\itshape}{\rmfamily}

\newtheorem{procedure}{Procedure}{\itshape}{\rmfamily}

\makeatletter
\@addtoreset{equation}{section}
\def\section{\@startsection {section}{1}{\z@}{-3.5ex plus -1ex minus
 -.2ex}{2.3ex plus .2ex}{\large\bf}}
\makeatother

\def\bfm#1{\mbox{\boldmath$#1$}}

\def\0{\bfm 0}

\newcommand{\bluecomment}[1]{\textcolor{black}{\textrm{#1}}}
\newcommand{\redcomment}[1]{\textcolor{black}{\textrm{#1}}}

\DeclareMathAlphabet{\mathpzc}{OT1}{pzc}{m}{it}

\newcommand{\de}{\backslash}

\def\bfm#1{\mbox{\boldmath$#1$}}

\begin{document}
%\pagestyle{empty}

%\mainmatter

\title{\bf Finding Connected Dense $k$-Subgraphs \thanks{Research supported in part by by NNSF of China under Grant No. 11222109, 11021161 and 10928102,
by 973 Project of China under Grant No. 2011CB80800,   and by CAS Program for Cross \& Cooperative Team of Science  \& Technology Innovation.}}

%\title{Finding Connected Dense $k$-Subgraphs}

 \author{Xujin Chen, Xiaodong Hu, Changjun Wang}

\date{Institute of Applied Mathematics, AMSS, Chinese Academy of Sciences\\ Beijing 100190, China\\
${}$\\
\mailsa}

\maketitle

%\vspace{-4mm}
\begin{abstract}
Given a connected graph $G$ on $n$ vertices and {a} positive integer $k\le n$, a subgraph of $G$ on $k$ vertices is {called} a $k$-subgraph in $G$. We design combinatorial approximation algorithms for finding a connected $k$-subgraph in $G$ such that its density is at least a factor  $\Omega(\max\{n^{-2/5},k^2/n^2\})$  of the density  of the densest $k$-subgraph in $G$ (which is not necessarily connected). These particularly provide the first non-trivial approximations for the densest connected $k$-subgraph problem on general graphs.
\end{abstract}

\noindent{\bf Keywords:} {Densest $k$-subgraphs, Connectivity, Combinatorial approximation algorithms
\newcounter{my}
\newenvironment{mylabel}
{
\begin{list}{(\roman{my})}{
\setlength{\parsep}{-0mm}
\setlength{\labelwidth}{8mm}
\setlength{\leftmargin}{8mm}
\usecounter{my}}
}{\end{list}}

\newcounter{my2}
\newenvironment{mylabel2}
{
\begin{list}{(\alph{my2})}{
\setlength{\parsep}{-1mm} \setlength{\labelwidth}{12mm}
\setlength{\leftmargin}{14mm}
\usecounter{my2}}
}{\end{list}}

\newcounter{my3}
\newenvironment{mylabel3}
{
\begin{list}{(\alph{my3})}{
\setlength{\parsep}{-1mm}
\setlength{\labelwidth}{8mm}
\setlength{\leftmargin}{10mm}
\usecounter{my3}}
}{\end{list}}

\section{Introduction}
Let $G=(V,E)$  be a connected simple undirected graph with $n$ vertices, {$m$ edges}, and nonnegative edge weights. The ({\em weighted}) {\em density}  of   $G $ is defined as its average (weighted) degree. Let $k\le n$ be a positive integer.
 A subgraph of $G$ is called a {\em $k$-subgraph} if it   has exactly $k$ vertices. The \emph{densest $k$-subgraph problem} (D$k$SP)   is to find a   $k$-subgraph of $G$ that has the maximum density, equivalently, a maximum number of  edges. If the $k$-subgraph requires to be connected, then the problem is referred as to the \emph{densest connected $k$-subgraph problem} (DC$k$SP).  Both D$k$SP and DC$k$SP have their weighted generalizations, denoted respectively as H$k$SP and HC$k$SP, which    ask  for a  heaviest (connected) $k$-subgraph, i.e., a (connected) $k$-subgraph with a maximum total edge weight.  Identifying $k$-subgraphs with high
{densities} is a useful primitive,  which   arises in diverse applications -- from social networks, to protein interaction
graphs, to the world wide web, etc. While dense subgraphs can give valuable information about
{interactions} in these networks, the additional connectivity requirement turns out to be natural in \bluecomment{various scenarios. %, {see e.g., \cite{ks09,lrja10}.}
One of typical examples is searching for a large community. If most vertices belong to  a dense connected subnetwork, only a few selected inter-hub links are needed to have a short average distance between any two arbitrary vertices in the entire network. Commercial airlines employ this hub-based routing scheme \cite{lrja10}.}

 \paragraph{Related work.} An easy  {reduction} from the maximum clique problem shows that D$k$SP, DC$k$SP and their weighted generalizations are all NP-hard in general. {The NP-hardness remains even for some very restricted graph classes
such as chordal graphs, triangle-free graphs, comparability graphs \cite{cp84} and bipartite graphs of
maximum degree  three \cite{fs97}. %Moreover, H$k$SP and  HC$k$SP remain NP-hard in the metric case, i.e., the edge weights of complete graphs satisfy  the triangle inequality \cite{rrt94}, and binary weighted case for cographs and split graphs~\cite{cp84}.
} %The  complexity status of D$k$SP on planar graphs and interval graphs has been open for three decades \cite{cp84}, while the NP-hardness of DC$k$SP on planar graphs has been established by {reducing} from the connected vertex cover problem \cite{kb91}.

 Most literature on finding dense subgraphs focus on the versions without requiring subgraphs to be connected.  For D$k$SP and its generalization H$k$SP, narrowing the large gap between the   lower and upper bounds on the approximabilty   is an important open problem.
 On the negative side,  the decision problem version of D$k$SP, in which one is asked if there is a $k$-subgraph with more than $h$ edges, is NP-complete even if $h$ is restricted by \redcomment{$h\le k^{1+\varepsilon}$ \cite{ai95}.}
 Feige \cite{f02} showed that computing a $(1+\varepsilon)$-approximation  for D$k$SP is at least as hard as refuting random 3-SAT clauses for some $\varepsilon>0$.
 Khot \cite{k06} showed that there does not exist any  polynomial time approximation scheme (PTAS) for D$k$SP assuming %NP $\not\subseteq$ $\cap_{\epsilon > 0}$ BPTIME($2^{n^\epsilon}$), i.e.,
 NP does not have randomized algorithms that run in sub-exponential time.  Recently, constant factor approximations in polynomial
time for   D$k$SP have been ruled out by Raghavendra and Steurel \cite{rs10} under Unique Games with Small
Set Expansion conjecture, and by Alon et al. \cite{aammw11} under certain ``average case'' hardness assumptions. On the positive side, considerable efforts have been devoted to finding good quality approximations for H$k$SP. Improving the $O(n^{0.3885})$-approximation of Kortsarz and Peleg \cite{kp93},   Feige et al. \cite{fkp01} proposed a combinatorial algorithm with approximation ratio $O(n^{\delta})$ for some  $\delta<1/3$. The latest algorithm of Bhaskara et al. \cite{bccfv10} provides an $O(n^{1/4+\varepsilon})$-approximation in $n^{O(1/\varepsilon)}$ time. If allowed to run for $n^{O(\log n)}$ time, their algorithm guarantees an approximation ratio of $O(n^{1/4})$. The $O(n/k)$-approximation algorithm by Asahiro et al. \cite{aitt00} is remarkable for its simple greedy removal method. {Linear and semidefinite programming (SDP) relaxation approaches have been adopted in \cite{fl01,hyz02,sw98} to design randomized rounding {algorithms},  where Feige and Langberg \cite{fl01} obtained an approximation ratio somewhat better than $n/k$, while the   algorithms of Srivastav and Wolf \cite{sw98} and Han et al. \cite{hyz02}   outperform this ratio   for a range of values $k=\Theta(n)$.} \bluecomment{On the other hand, the SDP relaxation methods have a limit of $n^{\Omega(1)}$ for D$k$SP as shown by Feige and Seltser \cite{fs97} and Bhaskara et al.~\cite{bcvgz12}.}

For some special  \bluecomment{cases in terms of graph classes,  values of $k$ and optimal objective values,}
better approximations have been obtained for  D$k$SP and H$k$SP.  Arora et al. \cite{akk95} gave a PTAS for the restricted  D$k$SP where $m=\Omega(n^2)$ and $k=\Omega(n)$, or each vertex of $G$ has degree $\Omega(n)$. Kortsarz and Peleg \cite{kp93} approximated D$k$SP with ratio $O((n/k)^{2/3})$ when the number of edges in the optimal solution is larger than \redcomment{$2\sqrt{k^5/n}$.}
 \mbox{Demaine et al. \cite{dhk05}} developed a 2-approximation algorithm for D$k$SP on $H$-minor-free graphs, where $H$ is any given fixed undirected graph. Chen et al. \cite{cfl11} showed  that D$k$SP on a large family of intersection graphs, including chordal graphs, circular-arc graphs and claw-free graphs, admits constant factor approximations. Several PTAS have been designed for D$k$SP on unit disk graphs \cite{cfl11}, interval graphs \cite{n11}, and a subclass of chordal graphs \cite{lmpz07}.

The work on approximating  densest/heaviest connected $k$-subgraphs are relatively very limited.  To the best of our knowledge, the existing polynomial time algorithms deal only  with special graphical topologies, including: (a)   4-approximation \cite{rrt94} and 2-approximation \cite{hrt97} for the metric H$k$SP and  HC$k$SP, where the underlying graph $G$ is complete, and the connectivity is trivial; (b) exact algorithms for      H$k$SP and  HC$k$SP on  trees \cite{cp84}, for D$k$SP and DC$k$SP on $h$-trees, cographs and split graphs \cite{cp84}, and  for
 DC$k$SP on interval graphs {whose   clique graphs are simple paths \cite{lmz05}.}

 Among the well-known relaxations of D$k$SP and H$k$SP is the problem of finding a   (connected) subgraph (without any cardinality constraint) of maximum weighted density. It is strongly polynomial time solvable using max-flow based techniques \cite{g84,l76}. Andersen and Chellapilla \cite{ac09} and Khuller  and Saha \cite{ks09} studied two relaxed variants of H$k$SP for finding a weighted densest subgraph with at least or at most $k$ vertices. The former variant was shown to be NP-hard even in {the} unweighted case, and admit 2-approximations in {the} weighted setting. The approximation of the latter variant was proved to be as hard as that of   D$k$SP/H$k$SP up to a constant factor. %, indicating that this problem is likely to be quite difficult.

\paragraph{Our results.} Given the interest in finding  densest/heaviest connected $k$-subgraphs from both
the theoretical and practical point of view,
  a better understanding of the {problems} is
an important challenge for the field. In this paper, we design {$O(mn\log n)$} time combinatorial approximation algorithms for finding a connected $k$-subgraph of $G$ whose density (resp. weighted density) is at least a factor $\Omega(\max\{n^{-2/5},k^2/n^2\})$ (resp.  $\Omega(\max\{n^{-2/3},k^2/n^2\})$) {of} the density (resp. weighted density) of the densest (resp. heaviest) $k$-subgraph of $G$ which is not necessarily connected.    These particularly provide the first non-trivial approximations for the densest/heaviest connected $k$-subgraph problem on general graphs: $O(\min\{n^{2/5},n^2/k^2\})$ for DC$k$SP and $O(\min\{n^{2/3},n^2/k^2\})$ for HC$k$SP.

{To evaluate the quality of our algorithms' performance guarantees $O(n^{2/5})$ and $O(n^{2/3})$, which are compared with the optimums of D$k$SP and H$k$SP, we investigate the maximum ratio  $\Lambda$ (resp.  $\Lambda_w$), over all graphs $G$ (resp. over all graphs $G$ and all nonnegative edge weights), between the maximum density (resp. weighted density) of {\em all} $k$-subgraphs and that of {\em all connected} $k$-subgraphs in $G$. The following examples show $\Lambda\ge n^{1/3}/3$ and $\Lambda_w\ge n^{1/2}/2$.}
\begin{example}\label{eg1}
(a) The graph $G$   is formed from $\ell$ vertex-disjoint  $\ell$-cliques $L_1,\ldots,L_\ell$ by adding, for each $i=1,\ldots,\ell-1$, a path $P_i$ of length $\ell^2+1$ {to connect $L_i$ and $L_{i+1}$,} where $P_i$ intersects all the $\ell$ cliques only at a vertex in $L_i$ and a vertex in $L_{i+1}$.
   Let $k=\ell^2$. Note that $G$ has $n=\ell^2+\ell^2(\ell-1)=\ell^3$ vertices. {The unique densest $k$-subgraph of $G$ is the disjoint union of $L_1,\ldots, L_{\ell}$ and has density $\ell-1$. One of densest connected $k$-subgraphs of $G$ is induced by the $\ell$ vertices in $L_1$ and certain $\ell^2-\ell$ vertices in $P_1$, and has density $(\ell(\ell-1)+2(\ell^2-\ell))/\ell^2$.  Hence $\Lambda\ge \ell^2/(\ell+2\ell)= n^{1/3}/3$.}

   (b)
The graph $G$ {is a tree} formed from a star on $\ell+1$ vertices by dividing each edge into a path of length $\ell+1$. All pendant edges have weight 1 and other edges have weight 0. Let $k=2\ell$. Note that $G$ has $n=\ell^2+1$ vertices. The unique heaviest $k$-subgraph of $G$ is induced by the $\ell$ pendant edges of $G$,   and has weighted density $1$.  Every heaviest connected $k$-subgraph {of $G$  is a path containing exactly one pendant edge of $G$, and   has weighted density $1/\ell$. Hence $\Lambda_w\ge\ell\ge n^{1/2}/2$.}\end{example}

\medskip
{The remainder of this paper is organized as follows. Section 2 gives notations, definitions and basic properties necessary for our discussion. Section 3 is devoted to designing approximation algorithms for finding connected dense $k$-subgraphs. Section 4 discusses extension to the weighted case, and future research directions.}
\section{Preliminaries}

Graphs studied in this paper are simple and undirected. {For any graph $G'=(V',E')$ and any  vertex $v\in V'$, we use $d_{G'}(v)$ to denote  $v$'s degree in $G'$.
The \mbox{{\em density}} ${\sigma}(G')$ of  $G' $ refers to its average degree, i.e. ${\sigma}(G')=\sum_{v\in V'}d_{G'}(v)/|V'|=2|E'|/|V'|$.} {Following convention, we define $|G'|=|V'|$. By a {\em component} of $G'$ we mean a maximal connected subgraph of $G'$.}

Throughout let $G=(V,E)$ be a connected graph on $n$ vertices and $m$ edges, and let $k\in[3,n]$ be an integer. Our goal is to find a connected $k$-subgraph $C$ of $G$ such that its density $\sigma(C)$ is as large as possible. %For some vertex $v\in V$, if deleting this vertex from $G$ increases its density, then we call it \emph{``deletable''}.
Let $\sigma^*(G)$ and  ${\sigma}_k^*(G)$   denote  the maximum densities of a subgraph and a $k$-subgraph   of $G$, respectively, {where the subgraphs are not necessarily connected}. It is clear that
\begin{equation}
{\sigma}^*(G)\ge {\sigma}_k^*(G)\text{ and }n-1\ge{\sigma}(G)\ge k\cdot {\sigma}_k^*(G)/n.\label{bound1} \end{equation}
%, and use $d_G(v)$ to denote the degree of a vertex $v$ in $G$.
Let $S$ be a subset of $V$ or a subgraph of $G$. We use $G[S]$ to denote the subgraph of $G$ induced by the vertices in $S$, and use $G\setminus S$ to denote the graph obtained from $G$ by removing all vertices in $S$ and their incident edges. If $S$ consists of {a single} vertex $v$, we write $G\setminus v$ instead of $G\setminus\{v\}$.
\begin{lemma} \label{lem1}
  ${\sigma}_k^*(G)<\sigma_{k-1}^*(G)+2$ and ${\sigma}_k^*(G)\le3\cdot\sigma_{k-1}^*(G)$.
\end{lemma}
\begin{proof}The first inequality in the lemma implies the second {since $\sigma_{k-1}^*(G)\ge1$}. %In case of $\sigma_{k-1}^*(G)=1$, the first inequality gives $\sigma_{k}^*(G)<3$, and $\sigma_{k}^*(G)\le2=2\cdot\sigma_{k-1}^*(G)$ as desired.\redcomment{(Problem: $\sigma_k^*(G)$ and $\sigma_{k+1}^*(G)$ are not necessary integer. we need a new proof.)}
To prove   ${\sigma}_k^*(G)<\sigma_{k-1}^*(G)+2$, consider   a densest $k$-subgraph $H$  of $G$, and  $v\in V(H)$. Then $d_H(v)\le k-1$, and  %${\sigma}_k^*(G)=\frac{\sum_{i=1}^kd_i}k\le2\cdot\frac{\sum_{i=1}^{k-1}d_i}{k-1}\le2\cdot\sigma_{k-1}^*(G)$.
 \begin{center}$\sigma_{k-1}^*(G)\ge\sigma(H\setminus v)=\frac{k\cdot\sigma(H)-2d_H(v)}{k-1}>\sigma(H)-\frac{2(k-1)}{k-1}={\sigma}_k^*(G)-2  $,\end{center}
establishing the lemma. \end{proof}

The vertices whose removals increase \bluecomment{the  density of the graph} play an important role in our algorithm design.
\begin{definition}\label{remove}
 A vertex $v\in V$ is called {\em removable} in $G$ if $\sigma(G\setminus v)>\sigma(G)$.
 \end{definition}
 Since $\sigma(G\setminus v)=2(|E|-d_G(v))/(|V|-1)$, the following is straightforward. It also provides an efficient way to identify removable vertices.
 %Let $A(G,k)$ denote the density of vertex induced subgraph returned by algorithm $A$ on input $(G,k)$.
\begin{lemma}\label{del}
{A vertex $v\in V$ is removable in $G$} if and only if $d_G(v)\!<\!{\sigma}(G)/2$.~
\end{lemma}
%\begin{proof} \[v \text{ is deletable}\Leftrightarrow\frac{2|E|}{|V|}<\frac{2(|E|-d(v))}{|V|-1}\Leftrightarrow d_G(v)<\frac{|E|}{|V|}=\frac{{\sigma}(G)}{2}.\] \end{proof}
\begin{lemma}\label{add}
Let $G_1$ be a connected $k$-subgraph of $G$. For any connected subgraph $G_2$ of $G_1$, it holds that   $\sigma(G_1)\geq \sigma(G_2)/\sqrt{k}$.
\end{lemma}
\begin{proof} Suppose that $G_2$ is a $k_2$-subgraph of $G$ with $m_2$ edges. By {the definition of density}, $\sigma(G_2)\le k_2-1$. The connectivity of $G_1$ implies $|E(G_1)|\ge|E(G_2)|+|V(G_1\setminus G_2)|$, and
\[ \sigma(G_1)\ge \frac{2(m_2+k-k_2)}{k}=\frac{k_2\cdot\sigma(G_2)+2(k-k_2)}{k}.
\]
In case of $k_2\ge\sqrt{k}$, we have {$\sigma(G_1)\ge k_2\cdot \sigma(G_2)/k\ge  \sigma(G_2)/\sqrt{k}$. In case of $k_2<\sqrt{k}$, since $k\ge3$, it follows that $G_1$ has no isolated vertices, and $\sigma(G_1)\ge1> k_2/\sqrt{k} >  \sigma(G_2)/\sqrt{k}$.}
\end{proof}

For a cut-vertex $v$ of $G$, we use $G_v$ to denote a densest component of  $G\setminus v$,  and use $G_{v+}$ to denote the {connected} subgraph of $G$ induced by $V(G_v)\cup\{v\}$.  {Note that} $G\setminus G_v$ is   a connected subgraph of $G$.

\section{Algorithms}
We design an $O(n^2/k^2)$-approximation algorithm (in Section 3.1) and further an $O(n^{2/5})$-approximation algorithm (in Section 3.2) for D$k$SP that always finds a connected $k$-subgraph of $G$. For ease of description we  assume $k$ is even. The case of odd $k$ can be treated similarly. Alternatively, if $k$ is odd, we can first find a connected $(k-1)$-subgraph $G_1$ satisfying $\sigma_{k-1}^*(G)/\sigma(G_1)\le O(\alpha)$, where $\alpha\in\{n^2/k^2,n^{2/5}\}$; it follows from Lemma \ref{lem1} that $\sigma^*_{k}(G)/\sigma(G_1)\le O(\alpha)$. Then we attach an appropriate  vertex to $G_1$, making a connected $k$-subgraph $G_2$ with {density} $\sigma(G_2)\ge \frac{k-1}k\sigma(G_1)\ge\frac23\sigma(G_1)$. This guarantees that the approximation ratio is still ${\sigma}_k^*(G)/\sigma(G_2)\le O(\alpha)$.

\subsection{$O(n^2/k^2)$-approximation}\label{k/n}
We first give an outline of our algorithm (see Algorithm \ref{alg1}) for finding a connected $k$-subgraph $C$ of $G$ with density {$\sigma(C)\ge \Omega(k^2/n^2)\cdot\sigma^*_k(G)$} (see Theorem \ref{the1}).

\paragraph{Outline.} We start with a connected graph $G'\leftarrow G$ and repeatedly delete removable vertices from $G'$ {to increase its density without destroying its connectivity.}
\begin{itemize}
\item If we can reach $G'$ with $|G'|=k$ in this way, we output $C$ as the resulting~$G'$.
\vspace{-2mm}\item If we can find a removable cut-vertex $r$ in $G'$ such that $|G'_r|\ge k$, then we recurse with $G'\leftarrow G'_r$.
\vspace{-2mm}\item If we stop at a $G'$ without any removable vertices, then we construct $C$ from an arbitrary connected $(k/2)$-subgraph by greedily attaching $k/2$ more vertices (see Procedure \ref{pro1}).
\vspace{-2mm}\item  If we are in none of  the above three cases,
 we find a connected subgraph  of $G'$ induced by a set $S$ of at most $ k/2$ vertices, and \bluecomment{then} expand the subgraph in two ways: (1) attaching $G'_r$ for all removable vertices $r$ of $G'$ which are contained in $S$, and (2) greedily attaching no more than $k/2$ vertices. From the resulting connected subgraphs, we choose the one that has more edges \bluecomment{(breaking ties arbitrarily)}, and further expand it to be a connected $k$-subgraph  (see Procedure \ref{pro2}), which is returned as the output~$C$.
\end{itemize}
\paragraph{Greedy attachment.} We describe how the greedy attaching mentioned in the above outline {proceeds}. Let $S$ and $T$ be disjoint nonempty  vertex subsets (or subgraphs) of $G$. Note that $1\le|S|<n$. The set of edges of $G$ with one end in $S$ and the other in $T$ is written as $[S,T]$. For any positive integer $j\le n-|S|$, a set $S^{\star}$ of $j$ vertices in $G\setminus S$ with {\em maximum} $|[ S, S^{\star}]|$ can be found greedily by  sorting the  vertices in $ G \setminus S $ as $v_1,v_2,\ldots,v_{j},\ldots$ in a non-increasing order of  the number of neighbors they have in $S$. For each $i=1,\ldots,j$, it can be guaranteed that $v_i$ has either a neighbor in $S$ or a neighbor in $\{v_1,\ldots,v_{i-1}\}$; in the latter case $i\ge2$.  Setting $S^{\star}=\{v_1,v_2,\ldots,v_j\}$. It is easy to see that
\begin{equation}\label{max}
\begin{array}{c}|[S,S^{\star}]|\ge \frac{j}n\cdot|[S,G\setminus S]| .\end{array}
\end{equation}
Moreover, if $G[S]$ is connected, the choices of $v_i$'s guarantee  that $G[S\cup S^{\star}]$ is connected. We refer to this $S^{\star}$ as a {\em $j$-attachment} of $S$ in $G$. Given $S$, finding a $j$-attachment of $S$ takes {$O(m+n\log n)$} time, which implies the following procedure runs in {$O( |E(G')|+|G'|\cdot\log |G'|)$} time.

\begin{procedure}\label{pro1}
{\em Input:}  {a} connected graph $G' $  without removable vertices, where  $| G'  | >  k$. \\
 {\em Output:} a connected $k$-subgraph of $G'$, written as {\sc Prc1}($G'$).
\end{procedure}
\vspace{-0mm}\hrule
\begin{enumerate}
\vspace{-2mm} \item  \vspace{0mm} $G_1=(V_1,E_1)\leftarrow$ an arbitrary connected $(k/2)$-subgraph of $G'$
\vspace{-2mm} \item  \vspace{0mm}   $V_1^{\star}\leftarrow$ a $(k/2)$-attachment of $V_1$ in $G'$  %Sort the  vertices in $ G' \setminus  G_1 $ as $v_1,v_2,\ldots,v_{k/2},\ldots$ in a non-increasing order of  the number of neighbors they have in $G_1$.
\vspace{-2mm} \item \vspace{0mm}%Let $V_2$ be the $k/2$ vertices in $V'\setminus V_1$ with the largest number of neighbors in $V_1$. Output the subgraph $G'[V_1\cup V_2]$.
 Output {\sc Prc1}$(G')\leftarrow$ $G[V_1\cup V_1^{\star}]$\vspace{-1mm} %$V(G_1)\cup\{v_1,v_2,\ldots,v_{k/2}\}$
\end{enumerate}
\vspace{-1mm}\hrule
\medskip
Note that the definition of attachment guarantees that $V_1\cap V_1^{\star}=\emptyset$, $|[V_1,V_1^{\star}]|$ is maximum, and $G[V_1\cup V_1^{\star}]$ is connected.
\begin{lemma}\label{lem3}
%If there is no removable vertex in $G'$, then the density of output subgraph satisfies
 $\sigma(\text{\sc Prc1}(G'))\ge \frac{k}{4|G'|}\cdot\sigma(G')$.
\end{lemma}
\begin{proof}
 Since   $G'$ has no removable vertices,   we deduce from Lemma \ref{del} that every vertex of $G'$ has degree at least ${\sigma}(G')/2$. Therefore {$|[G_1,   G'\setminus G_1]|\ge \frac{k}2\cdot\frac{\sigma(G')}2-2|E_1|$.} Recalling  (\ref{max}), % the choices of $v_1,v_2,\ldots,v_{k/2}$,
 we see that the number of edges in $\text{\sc Prc1}(G') $ is at least $ |[V_1,V_1^{\star}]|\ge{(\frac{k\cdot{\sigma}(G')}4-2|E _1 |)}\cdot\frac{k/2}{| G' |}+|E_1|\ge\frac{ k^2}{8|G'|}\cdot{\sigma}(G') $,  proving the lemma.   % So ${\sigma}(H)\ge 2k^2{\sigma}(G')/8|V'|k\ge k\sigma(G')/4|V'|$.
\end{proof}
 %Besides, we use $U$ to denote a set of some cut-vertices of $G$.

\begin{procedure}\label{pro2}
{\em Input:} {a} connected graph $G'$ with $|G'|>k$, where every removable vertex $r$ is a cut-vertex satisfying $|G'_r|<k$. \quad
{\em Output:} a connected $k$-subgraph of $G'$, written as  {\sc Prc2}($G'$).
\end{procedure}
\vspace{-0mm}\hrule
\begin{enumerate}
 \vspace{-2mm} \item  \vspace{0mm} $H\leftarrow G'$,   {$R'\leftarrow R=$ the set of removable vertices of $G'$}

 \vspace{-2mm}   \item  \vspace{0mm}  \textbf{While} $R'\neq\emptyset$  \textbf{do}

\vspace{-2mm}  \item \vspace{0mm} \hspace{6mm}    Take $r\in R'$
\vspace{-2mm}    \item \vspace{0mm} \hspace{6mm} $H\leftarrow H\setminus V(G'_r)$,     \   $R'\leftarrow R'\setminus V(G'_{r+})$ \label{delete}     %\cap V(H)
\vspace{-2mm}    \item  \vspace{0mm}  \textbf{End-While} \label{end}

\vspace{-2mm}  \item  \vspace{0mm}  For each $v\in V(H)$,  define  $\theta(v)=|G'_{v+}|$ if $v\in R$, and $\theta(v)=1$ otherwise\label{size}

\vspace{-2mm}  \item \vspace{0mm} Let $S$ be a {\em minimal} subset of $V(H)$ s.t. \!$H[S]$ is connected \&    $\sum_{v\in S}\theta(v) \!\ge \! \frac{k}2$ \label{h1} %Start from any vertex in $G''$, select some connected vertices along a path such that the sum of its weights just exactly pass $k/2$, denote them by $H_1$;
\vspace{-2mm}  \item  \vspace{0mm} %Sort the   vertices of $H\setminus H_1$ in a non-increasing order of the number of neighbors they have in $H_1$.
   Let $S^{\ast}$ be a   $\min\{k/2,|H\setminus S|\}$-attachment of $S$ in $H$\label{step2} %vertices in $H\setminus H_1$ such that $[H_1,S_2]|$ is maximized and $H[V(H_1)\cup S_2]$ is connected % the \redcomment{subgraph induced on} $k/2$ vertices in $H\setminus H_1$ with the largest number of neighbors in $H_1$ \redcomment{(If there are less then $k/2$ vertices, let $H_2=H\setminus H_1$)};
 \vspace{-2mm} \item  \vspace{0mm} {$V_1\leftarrow S\cup(\cup_{r\in R\cap S}V(G'_r))$,\; $V_2\leftarrow S\cup S^{\star}$}
 \vspace{-2mm} \item \vspace{0mm} Let $H'$ be one of    $G'[V_1]$ and $G'[V_2]$  {whichever has more edges (break ties arbitrarily)}\label{more}
\vspace{-2mm}  \item \vspace{0mm} Expand $H'$ to be a connected $k$-subgraph of $G'$\label{exp}
\vspace{-2mm}  \item  \vspace{0mm} Output {{\sc Prc2}$(G')\leftarrow H'$}
\end{enumerate}
\vspace{-1mm}\hrule
\medskip
Under the condition that the resulting graph is connected, the expansion in Step~\ref{exp} can be done in an arbitrary way. It is easy to see that Procedure \ref{pro2} runs in $O(|G'|\cdot|E(G')|)$ time.
\begin{lemma}\label{lem4}
At the end of the while-loop (Step \ref{end}) in Procedure \ref{pro2},  we have
\vspace{-2mm}\begin{mylabel}
\item $H$ is a connected  subgraph of $G'$.
\item  If  $H$ contains two distinct vertices $r$ and $s$ that are removable in $G'$, then (by the condition of the procedure both $r$ and $s$ are cut-vertices of $G'$, {and moreover}) {$G'_r$ and $ G'_s$} are vertex-disjoint.
\end{mylabel}
\end{lemma}
\begin{proof} Note that in every execution of the while-loop,   $r\in R'$ is    a cut-vertex of $H$, and $V(H)\cap V(G'_r)$ induces  a component of $H\setminus r$. Thus $H$ is connected throughout the procedure. For any two removable vertices $r,s $ of $G'$ with $|G'_r|\le|G'_s|$ and $r,s\in V(H)$, if $G'_r$ and $G'_s$ are not vertex-disjoint, then $V(G'_r)\cup\{r\}\subseteq V(G'_s)$. It  follows that all vertices of $V(G'_r)\cup\{r\}$ have been removed by Step \ref{delete} when considering $s\in R'$, a contradiction.   %  thus $G''\setminus G'_u$ is  connected. At step \ref{end}, suppose on the contrary that there exists two distinct $u,w\in U\cap G''$ and $G'_u\cap G'_w\neq\emptyset$. Since in graph $G'$, $G'_u$ is only connected to cut-vertex $u$ and $G'_w$ is only connected to cut-vertex $w$, so it must be one of the cases that $u\in G'_w$ or $w\in G'_u$. Without loss of generality, suppose $w\in G'_u$, then $w$ will be deleted from $G''$ when the while-loop deleting $G'_u$, it's a contradiction to the fact that $u$ and $w$ are both in the final $G''$.
\end{proof}
\bluecomment{Observe that for any two distinct $r,s\in R$, either $G'_{r+}$ and $G'_{s+}$ are vertex-disjoint, or $G'_{r+}$ contains $G'_{s+}$,  or $G'_{s+}$ contains $G'_{r+}$. This fact, along with an inductive argument, shows that, throughout Procedure \ref{pro2}, for any $s\in R\de V(H)$, there exists at least a vertex $r\in V(H)\cap R$ such that $G'_{r+}$ contains $G'_{s+}$, implying that $(U_{r\in R\cap V(H)}V(G_{r+}))\cup(V(H)\de R)=V(G')$ holds always. By Lemma \ref{lem4}(ii), in Step \ref{h1}, we see that $V(G')$ is the disjoint union of $V(G_{r+})$, $r\in R\cap V(H)$ and $V(H)\de R$, giving $\sum_{v\in V( H)}\theta(v)=|G'|>k$. Hence, the connectivity of $H$ (Lemma \ref{lem4}(i)) implies that the set $S$ at Step \ref{h1} does  exist.}

 \bluecomment{Take $u\in S$ such that $u$ is not a cut-vertex of $H$. If $|S|\ge(k/2)+1$,  then   we have $\sum_{v\in S\de\{u\}}\theta(v)\ge|S\de\{u\}|\ge k/2$, a contradiction to the minimality of $S$. Hence
 \begin{equation*}|S|\le k/2.\end{equation*}
  Since Step 4 has removed from $H$ all vertices in $V(G'_r)$ for all $r\in R$, we see that $V_1$ is the disjoint union of $S$ and $\cup_{r\in R\cap S}V(G'_r)$   Recall that $|G'_r|<k$   for all $r\in R\cap S$. If $|V_1|>k$, then $|S|\ge2$, and either $\theta_u\ge k/2$ or $ \sum_{v\in S\de\{u\}}\theta(v)\ge k/2$, contradicting to the minimality of $S$. Noting that $|V_1|=\sum_{v\in S}\theta(v)$, we have}    \begin{equation}\label{v1}
 k/2\le|V_1|\le k.\end{equation} We deduce that the output of Procedure \ref{pro2} is indeed a connected $k$-subgraph of~$G'$.
\begin{algorithm}\label{alg1}
{\em Input:} connected graph $G=(V,E)$ with $|V|\ge k$.\\
{\em   Output:}  a  connected $k$-subgraph of $G$, written as {\sc Alg1}$(G)$.
\end{algorithm}
\vspace{-0mm}\hrule
\begin{enumerate}
 \vspace{-2mm} \item  $G'\leftarrow G$

  \vspace{-2mm} \item  \vspace{0mm} \textbf{While} $|G'|>k$    {and $G'$ has a removable vertex $r$ that is not a cut-vertex} %$\exists$ non-cut vertex $v\in V(G')$ being removable,
       \textbf{do}\label{condition}

   \vspace{-2mm}\item \vspace{0mm} \hspace{4mm}  $G'\leftarrow G'\setminus r$
   \vspace{-2mm}\item \vspace{0mm} \textbf{End-While}\label{note} {\hfill{\small // either $|G'|=k$ or any removable vertex of $G'$ is a cut-vertex}}

   \vspace{-2mm}\item \vspace{0mm}  \textbf{If} $|G'|=k$  {{\bf then} output {\sc Alg1}$(G)\leftarrow G'$} \label{s1}

     \vspace{-2mm}  \item  \vspace{0mm}  {\textbf{If} $|G'|>k$ and $G'$ has no removable vertices

      \hspace{4mm}{\bf then} output  {\sc Alg1}$(G)\leftarrow$ {\sc Prc1}($G'$)} \label{s2}

    \vspace{-2mm}   \item   \vspace{0mm} {\textbf{If} $|G'|>k$ and $|G'_r|< k$ for each removable vertex $r$ of $G'$

      \hspace{4mm}{\bf then} output  {\sc Alg1}$(G)\leftarrow$ {\sc Prc2}($G'$)}  \label{s3}

    \vspace{-2mm} \item  \vspace{0mm}  {\textbf{If} $|G'|>k$ and $|G'_r|\ge k$ for some removable vertex $r$ of $G'$

    \hspace{4mm}{\bf then} output  {\sc Alg1}$(G)\leftarrow$ {\sc Alg1}($G'_r$)} \label{s4}
\end{enumerate}
\vspace{-2mm}\hrule
\vspace{3mm}

In the  {while-loop}, we repeatedly delete removable non-cut  vertices from $G'$  until $|G'|=k$ or  {$G'$ has} no removable non-cut vertex anymore.  {The deletion process keeps $G'$   connected, %its vertex degrees nonincreasing,
and its density  $\sigma(G')$} increasing (cf. Definition \ref{remove}). When the  {deletion process finishes, there are four possible cases, which are handled by Steps \ref{s1}, \ref{s2}, \ref{s3} and \ref{s4}, respectively.
\begin{itemize}
\vspace{-2mm}\item In case of Step \ref{s1}, the output $G'$ is clearly a connected $k$-subgraph of $G$.
\vspace{-2mm}\item In case of Step \ref{s2},  $G'$ qualifies to be an input of Procedure \ref{pro1}. With this input, Procedure \ref{pro1} returns the connected $k$-subgraph  {\sc Prc1}$(G')$  of $G'$ as the algorithm's output.
\vspace{-2mm}\item In case of Step \ref{s3},  $G'$ qualifies to be an input of Procedure \ref{pro2}. With this input, Procedure \ref{pro2} returns the connected $k$-subgraph  {\sc Prc2}$(G')$  of $G'$ as the algorihtm's output.
\vspace{-2mm}\item In case of Step \ref{s4},  the algorithm recurses with {smaller} input $G'_r$, which satisfies $\sigma(G'_r)\ge\sigma(G')\ge\sigma(G)$ and $k\le|G'_r|<|G'|\le|G|$.
\end{itemize}
Hence after $O(n)$ recursions, the algorithm terminates at one of Steps \ref{s1} -- \ref{s3}, and outputs a connected $k$-subgraph of $G$.}

% The algorithm runs Procedure \ref{pro1} or Procedure \ref{pro2} at most once, which takes  $O(mn)$ time. Using appropriate data structure, the removability of a vertex can be checked in $O(1)$ time at Step \ref{condition}. Once a vertex is found non-removable, it will never be checked by the subsequent implementations of the while-loop (since the vertex degree keeps nonincreasing and graph density keeps increasing, which implies that the vertex remains non-removable henceforth). It takes $O(m)$ time for us to determine whether a removable vertex $r$ of $G'$ is a cut-vertex of $G'$, and obtain $G'_r$ if it is. If $r$ is not a cut-vertex, then we remove $r$ from $G'$ and update $G'$ in $O(n)$ time. If $r$ is a cut-vertex with $|G'_r|<k$, then $r$ remains a cut-vertex of $G'$ in the subsequent process (note $|G'|\ge k$ holds always) unless it is removed from the graph by certain recursion  at Step \ref{s4}; so the subsequent while-loops never consider  it. If $r$ is a cut-vertex with $|G'_r|\ge k$, then we recurse on $G'_r$, throwing away $G'\setminus G'_r$ which contains $r$. Overall, the algorithm runs in $O(mn)$ time.
\begin{theorem}\label{the1}
Algorithm \ref{alg1} finds in  {$O(mn)$
 time} a connected $k$-subgraph   $C$ of $G$ such that $ {\sigma}_k^*(G)/\sigma(C)\le12n^2/k^2 $.
\end{theorem}
\begin{proof} Let $C=\text{\sc{Alg\ref{alg1}}}(G)$ be the output connected $k$-subgraph of $G$.
If $C$ is output at Step \ref{s1}, then  its density is ${\sigma}(C)\ge {\sigma}(G)\ge (k/n)\cdot{\sigma}_k^*(G)$, {where the last inequality is by (\ref{bound1}).} If  $C$ is output by Procedure \ref{pro1} at Step \ref{s2}, then from Lemma \ref{lem3} we know its density is at least $ \frac{k}{4|G'|}\cdot\sigma(G')\ge \frac{k}{4n}\cdot\sigma(G)\ge \frac{k^2}{4n^2}\cdot{\sigma}_k^*(G)$.

 Now we are only left with the case that $C =\text{\sc Prc2}(G')$ is output by Procedure~\ref{pro2} at Step \ref{s3} of Algorithm~\ref{alg1}. Let $R$ denote the set of removable vertices of $G'$. For every $r\in R$, we see that  $r$ is a cut-vertex of $G'$ (cf. the note at Step   \ref{note} of the algorithm), and ${\sigma}(G'_r)\ge{\sigma}(G'\setminus r)> {\sigma}(G')$, where the first inequality is from the definition of $G'_r$ (it is the densest component of $G'\setminus r$), and the second inequality is due to the removability of {$r$.} Thus  \[{\sigma}(G'_{r+})> \sigma(G'_r)\cdot|G'_r|/(|G'_r|+1)\ge {\sigma}(G')/2\;\text{ for every }r\in R.\] Using the notations in Procedure \ref{pro2}, we note that  each vertex of  $S\setminus R$ is non-removable in $G'$, and therefore has degree at least $\sigma(G')/2$ in $G'$ by Lemma \ref{del}. Since $V_1= S\cup(\cup_{r\in R\cap S}V(G'_r))=(S\setminus R)\cup(\cup_{r\in S\cap R}V(G'_{r+}))$ contains at least $k/2$ vertices {(recall (\ref{v1}))}, it follows that $G'$ contains at least  {$(\frac{k}2\cdot\frac{\sigma(G')}2)/2\ge \frac{k}8\cdot\sigma(G)\ge \frac{k^2}{8n}\cdot{\sigma}_k^*(G)$} edges each with at least one end in $V_1$. %, a vertex is either a cut-vertex $u\in U$ or a non removable vertex and $H'=H_1\bigcup_{u\in H_1\cap U}G'_u$, so  the number of edges of $G'$ with at least one endpoint in $H'$ is at least $k{\sigma}(G')/4$.

If there are at least {$ \frac{k^2}{24n}\cdot{\sigma}_k^*(G)$ edges with both ends in $V_1$, then by Step~\ref{more} of Procedure \ref{pro2} we have $|E(C)|\ge \frac{k^2}{24n}\cdot{\sigma}_k^*(G)$ and ${\sigma}(C)=2|E(C)|/k\ge \frac{k}{12n}\cdot{\sigma}_k^*(G) \ge    \frac{k^2}{12n^2}\cdot{\sigma}_k^*(G)$}. It remains to  consider the case where $G'$ contains at least   {$\frac{k^2}{12n}\cdot{\sigma}_k^*(G)$} edges between $V_1$ and $G'\setminus V_1$. All these edges are between $S$ and $G'\setminus V_1=H\setminus S$, since each edge  incident with  any vertex in $G'_r$ ($r\in R$) must have both ends in $V_1$. %, we know the edges with only one end-point in $H'$ is only connected to $H_1$.
 So, by the definition of $S^{\star}$ at Step  \ref{step2} of Procedure \ref{pro2}, we deduce from (\ref{max}) that there are at least {a number   $|[S,S^{\star}]|\ge \frac{k/2}n\cdot|[S,H\setminus S]|\ge\frac{ k^3}{24n^2}\cdot{{\sigma}_k^*}(G)$ of} edges in the subgraph of $G'$ induced by $V_2=S\cup S^{\star}$. Hence  {$ \sigma(C)\ge 2|[S,S^{\star}]|/k\ge \frac{k^2}{12n^2}\cdot{\sigma}_k^*(G) $}, \bluecomment{justifying the performance of the algorithm.}

 Algorithm \ref{alg1}  runs Procedure \ref{pro1} or Procedure \ref{pro2} at most once, which takes  $O(mn)$ time. {At least one of Procedures \ref{pro1} and \ref{pro2} has never been called by the algorithm}. Using appropriate data structures {and $O(n^2)$ time preprocessing, we construct a list $L$ of removable vertices in $G'$ (cf. Lemma \ref{del}). It takes $O(m)$ time for Step 2 to determine whether a removable vertex $r\in L$ is a cut-vertex of $G'$, and obtain $G'_r$ if it is. If $r$ is not a cut-vertex, then we remove $r$ from $G'$, and update $G'$ and $L$ in $O(n)$ time. If $r$ is a cut-vertex with $|G'_r|<k$, then $r$ remains a cut-vertex of $G'$ in the subsequent process (note $|G'|\ge k$ holds always) unless it is removed from the graph by certain recursion  at Step \ref{s4}; so the subsequent while-loops {will} never consider  it. If $r$ is a cut-vertex with $|G'_r|\ge k$, then we recurse on $G'_r$, and update $G'\rightarrow G'_r$ and $L$ in {$O(| G'_r|)=O(n)$} time, throwing away $G'\setminus G'_r$ which contains $r$. Overall, the algorithm runs in $O(mn)$ time.}
\end{proof}

\subsection{$O(n^{2/5})$-approximation}\label{2/5}
 In this subsection we design algorithms for finding connected $k$-subgraphs of $G$ that jointly provide an   $O(n^{2/5})$-approximation to D$k$SP. Among the outputs of all these algorithms (with input $G$), we select the densest one, denoted as $C$. Then it can be guaranteed that {$\sigma^*_k(G)/\sigma(C)\le O(n^{2/5})$}. In view of the $O(n^2/k^2)$-approximation of Algorithm \ref{alg1}, we may focus on the case of  $k< n^{4/5}$. (Note  that $ n^2/k^2 \le n^{2/5}$ if $k\ge n^{4/5}$.)

Let $D$ be a densest connected subgraph of $G$, which    is computable in time $O(mn\log(n^2/m))$   \cite{g84,l76} (because every component of a densest subgraph of $G$ is also a densest subgraph of $G$). Thus
\begin{center}${\sigma}(D)={\sigma}^*(G)\ge\sigma^*_k(G).$\end{center}
 Moreover, the maximality of $\sigma(D)$ implies that $D$ has no removable vertices.
\begin{algorithm}\label{pro3}
{\em Input:} connected graph $G$ along with its densest connected subgraph $D$.\\
{\em Output:} a connected $k$-subgraph of $G$, denoted as {\sc Alg\ref{pro3}}($G$).
\end{algorithm}
\vspace{-0mm}\hrule
\begin{enumerate}
%\vspace{0mm} \item  Find  a  densest connected subgraph $D$ of $G$

 \vspace{-2mm} \item  \textbf{If} $|D|\le k$   \textbf{then} Expand $D$ to be a connected $k$-subgraph $H$ of $G$

 \vspace{-1mm}\hspace{25mm} Output  {\sc Alg\ref{pro3}}$(G)\leftarrow H$

 \vspace{-2mm} \item \hspace{3mm}  \textbf{Else} Output  {\sc Alg\ref{pro3}}$(G)\leftarrow $ {\sc Prc\ref{pro1}}($D$)
\end{enumerate}
\vspace{-2mm}\hrule
\medskip

\begin{lemma}\label{per3}
If $k< n^{4/5}$, then
{${\sigma}(\text{\sc Alg\ref{pro3}}(G))\ge \min\{k/(4n),n^{-2/5}\}\cdot{\sigma}^*(G)$.}
\end{lemma}
 \begin{proof}In case of $|D|\le k$,   by Lemma \ref{add}, it follows from ${\sigma}^*(G)\ge {\sigma}_k^*(G)$  that the density of the output subgraph {${\sigma}(H)\ge \sigma(D)/\sqrt{k}={\sigma}^*(G)/\sqrt k
 $. Since $k\le n^{4/5}$,   we see that  $\sigma(H)\ge n^{-2/5}\cdot\sigma^*(G)$.} % approximates the densest $k$-subgraph within ratio    $\sqrt k\le n^{2/5}$.

  In case of $|D|>k$, we deduce from Lemma \ref{lem3} that  the connected $k$-subgraph {\sc Alg\ref{pro3}}($G$)={\sc Prc\ref{pro1}}($D$) of $D$ has density at least {$\frac{k}{4|D|}\cdot \sigma(D) \ge \frac{k}{4n} \cdot{\sigma}^*(G)$.} \end{proof}

  Our next algorithm is simply an expansion of Procedure 2 by Feige et al.~\cite{fkp01}.
{Let $V_h$ be a set of $k/2$ vertices of highest degrees in $G$, and let $d_h=\frac2k\sum_{v\in V_h}\!d_G(v)$ denote the average degree of {the} vertices in $V_h$.}
\begin{algorithm}\label{pro4}
{\em Input:} connected graph $G$ with $|G|\ge k$. \\
{\em Output:} a connected $k$-subgraph of $G$, denoted as {\sc Alg\ref{pro4}}($G$)..
\end{algorithm}
\vspace{-0mm}\hrule
\begin{enumerate}
 %\item    {$V_h\leftarrow$  a set of $k/2$ vertices with highest degrees in $G$ (breaking ties arbitrarily)} %Sort the remaining vertices by the number of neighbors they have in $L$. Let
 \vspace{-1mm}\item  \vspace{0mm} $V_h^{\star}\leftarrow$ a   $(k/2)$-attachment of $V_h$ in $G$
 \vspace{-2mm}\item  \vspace{0mm} {$H\leftarrow$ a densest component of $G[V_h\cup V_h^{\star}]$}%the connected  component of $G[V_h\cup V_h^{\star}]$ with maximum density.
%\item \vspace{0mm} {Expand $H$ to   be a connected $k$-subgraph $H'$ of $G$}
 \vspace{-2mm}\item \vspace{0mm} {Output   {\sc Alg\ref{pro4}}$(G)\leftarrow $ a $k$-connected subgraph of $G$ that is expanded from $H$}
\end{enumerate}
\vspace{-1mm}\hrule

\medskip
In the above algorithm,  the subgraph $G[V_h\cup V_h^{\star}]$  is exactly the output of Procedure 2 in~\cite{fkp01}, for which it  has been shown (cf, Lemma 3.2 of \cite{fkp01}) that
 \[\bar\sigma:=\sigma(G[V_h\cup V_h^{\star}])\ge kd_h/(2n).\] Together with  Lemma \ref{add}, we have the following result.
\begin{lemma}\label{expand}
%The density of subgraph $H$ outputted by Procedure \ref{pro4} satisfying
 {${\sigma}(\text{\sc Alg\ref{pro4}}(G))\geq \frac{\bar\sigma}{\sqrt k}\geq\frac{\sqrt{k}}{2n}\cdot d_h $.}
\end{lemma}
 \begin{proof} It follows from Lemma \ref{add} that ${\sigma}(\text{\sc Alg\ref{pro4}}(G))\ge\sigma(H)/\sqrt{k}\ge\bar\sigma/\sqrt{k}$.
 \end{proof}

Our last algorithm is a slight modification of Procedure 3 in \cite{fkp01}, where we link things up via a ``hub'' vertex.
For vertices $u,v$ of $G$, let $W(u,v)$ denote the number of walks of length $2$ from $u$ to $v$ in $G$.
\begin{algorithm}\label{pro5}
{\em Input:} connected graph $G=(V,E)$ with $|G|\ge k$. \\
{\em Output:} a connected $k$-subgraph of $G$, denoted as {\sc Alg\ref{pro5}}($G$).
\end{algorithm}
\vspace{-0mm}\hrule
\begin{enumerate}
\vspace{-1mm}
\item \vspace{0mm}  $G_\ell\leftarrow G[ V\setminus V_h]$.%, where $V_h$ is the set of $k/2$ vertices of highest degree in $G$, as defined in procedure \ref{pro4}.
 \vspace{-2mm}\item \vspace{0mm} Compute $W(u,v)$ for all pairs of vertices  {$u,v$} in $G_\ell$.
 \vspace{-2mm} \item \vspace{0mm} For every   $v\in V\setminus V_h$, construct a {connected $k$-subgraph  {$C^v$}  of $G$} as follows:
 \begin{itemize}
 \item[-] \vspace{-2mm} Sort the vertices   {$u\in V\setminus V_h\setminus\{v\}$} with positive $W(v,u)$  as $v_1,v_2,\ldots,v_t$ such that $W(v,v_1)\ge W(v, v_2)\ge \cdots\ge W(v,v_t)>0$.%, where \mbox{$t\le|V\setminus V_h|-1$.}
 \item[-] \vspace{-1mm} $P^v\leftarrow\{v_1,\ldots, v_{\min\{t,k/2-1\}}\}$
  \item[-] \vspace{-1mm}  {$B^v\leftarrow $ a set of $\min\{d_{G_\ell}(v),k/2\}$ neighbors of $v$ in $G_\ell$ such that {the number of edges between $B^v$ and $P^v$ is maximized}.} % In $G_\ell$, compute for every neighbor $x$ of $v$ the number of edges connecting $x$ to $P^v$, $deg(x,P^v)$, and construct a set $B^v$ containing the $k/2$ neighbors of $v$ with highest $deg(x,P^v)$. Let
  \item[-] \vspace{-1mm} $C^v\leftarrow$ {the component of $ G_\ell[\{v\}\cup B^v\cup P^v]$ that contains $v$} %\hfill{\small// ``hub'' vertex $v$ ensures that $C^v$ is connected} %without isolated vertices. (If $C^v$ still contains less than $k$ vertices, then it is completed to size $k$ by adding some connected vertices.)
  \item[-]\vspace{-1mm} Expand $C^v$ to be a connected $k$-subgraph of $G$
     \end{itemize}
\item \vspace{-2mm} Output {\sc Alg\ref{pro5}}$(G)\leftarrow$ the densest  $C^v$ for $v\in V\setminus V_h$
\end{enumerate}
\vspace{-1mm}\hrule

\medskip
In the above algorithm, $B^v$ can be   {found in $O(m+n\log n)$} time, and $v$ is the ``hub'' vertex ensuring that $C^v$ is connected.
Hence the algorithm is correct, and runs in $O(mn+n^2\log n)$ time, {where Step 2 finishes in $O(n^2\log n)$ time}.  The key point here is that $C^v$ contains all edges between  $B^v$ and $P^v$, {where $B^v$ and $P^v$ are not necessarily disjoint.}
 Using a similar analysis   to that in \cite{fkp01}, we obtain the following.
 % A similar analysis as \cite{fkp01} will give  the following result.
\begin{lemma}\label{per5}
 If {$k\!\le \!\frac23n$, then} $\sigma( \text{\sc Alg\ref{pro5}}(G))\ge  \frac{({\sigma}^*_{k}(G)-2\bar\sigma)^2}{2\max\{k, 2d_h\}}\!\cdot\!\frac{k-2}{k}\geq \frac{({\sigma}^*_{k}(G)-2\bar\sigma)^2}{6\max\{k, 2d_h\}}$.~
%Let $C$ be the output of procedure \ref{pro5}, then $C$ is connected and \[{\sigma}(\text {\sc Alg\ref{pro5}}(G))\ge \frac{({\sigma}_k(G)-kd_h/n)^2}{2\max\{k, 2d_h\}}\cdot\frac{k-1}{k} \] ,where $d_h$ is the average degree of the $k/2$ vertices of highest degree in $G$.
\end{lemma}
\begin{proof} From Lemma 3.3 of \cite{fkp01} we know that   $G_\ell$ contains a $k$-subgraph, denoted as $H$, with average degree at least $\sigma_k^*(G)-2\bar\sigma$. Note that the number of  length-2 walks within   $H$ is at least $k(\sigma_k^*(G)-2\bar\sigma)^2$. This is because each $v\in V( H)$ contributes $(d_H(v))^2$ to this number, and $\sum_{v\in V(H)}(d_H(v))^2\geq k(\sigma_k^*(G)-2\bar\sigma)^2$ by convexity. It follows that there is a vertex $v\in V(H)$ which is the endpoint of at least {a number $(\sigma_k^*(G)-2\bar\sigma)^2$ of} length-2 walks in $H$. By the  construction of $P^v$,  there are at least $(\sigma_k^*(G)-2\bar\sigma)^2\cdot\frac{(k/2-1)}{k}$ walks of length 2 between this {vertex} $v$ and  vertices in {$P^v$}. Therefore, the number of edges between $B^v$ and $P^v$ is at least $\frac{(\sigma_k^*(G)-2\bar\sigma)^2(k-2)}{2k}$  if  $d_{G_\ell}(v)\le k/2$, and at least $\frac{(\sigma_k^*(G)-2\bar\sigma)^2 (k-2)}{2k}\cdot\frac{k/2}{d_{G_\ell}(v)}$ edges {otherwise}. Since we do not require $P^v$ and $B^v$ to be disjoint, each edge may have been counted twice. Notice from the definition of $d_h$ that $d_{G_{\ell}}(v)\le d_G(v)\le d_h$. Since $C^v$ contains all edges between $B^v$ and $P^v$, it contains at least $\min\{\frac{(\sigma_k^*(G)-2\bar\sigma)^2(k-2)}{4k}, \frac{(\sigma_k^*(G)-2\bar\sigma)^2(k-2)}{8d_h}\}$ edges.
This guarantees $\sigma( \text{\sc Alg\ref{pro5}}(G))\ge  \frac{({\sigma}^*_{k}(G)-2\bar\sigma)^2}{2\max\{k, 2d_h\}}\cdot\frac{k-2}{k}$. {Since $k\ge3$, the lemma follows.}
\end{proof}
\iffalse
\begin{proof}
\redcomment{It has been shown in \cite{fkp01} that
$\max\{\sigma(C^v\setminus v): {v\in V(G_\ell)}\}\ge   \frac{({\sigma}_{k-1}(G)-(k-1)d_h/n)^2}{2\max\{k-1, 2d_h\}}$. Therefore
\[\sigma( \text{\sc Alg\ref{pro5}}(G))\ge\frac{({\sigma}_{k-1}(G)-(k-1)d_h/n)^2}{2\max\{k-1, 2d_h\}}\cdot\frac{k-1}{k}\ge \frac{({\sigma}_{k-1}(G)-kd_h/n)^2}{4\max\{k, 2d_h\}}
\] as desired}
\end{proof}
\fi

We are now ready to prove that the four algorithms given above  jointly guarantees an $O(n^{2/5})$-approximation.

\begin{theorem}
A connected $k$-subgraph $C$ of $G$ can be found in {$O(mn\log n)$} time such that $\sigma^*_k(G)/\sigma(C)\le O(n^{2/5})$.
%${\sigma}(H)\ge {\sigma}_k(G)/(4n^{2/5})$,  when  $k\le n^{4/5}$.
\end{theorem}
\begin{proof}
Let $C$ be the  densest connected $k$-subgraph of $G$ among the outputs of Algorithms \ref{alg1}   -- \ref{pro5}. As mentioned at the beginning of Section \ref{2/5}, it suffices to consider the case of $k<n^{4/5}$. The connectivity of $C$ gives $\sigma(C)\ge1$. {Clearly, we may assume $n\ge8$, which along with $k<n^{4/5}$ implies $k\le2n/3$.} By Lemmas \ref{per3} -- \ref{per5}, we may assume that %densest output of the above three  procedures and is connected, so we have
\[{\sigma}(C)\ge \max\left\{1, \text{ }\frac{k{\sigma}^*(G)}{4n},\text{ }\frac{\bar\sigma}{\sqrt{k}},\text{ } \frac{\sqrt{k}d_h}{2n}, \text{ } \frac{({\sigma}_k(G)-2\bar\sigma)^2}{6\max\{k, 2d_h\}} \right\}.\]
%\[{\sigma}(C)\ge \max\left\{1, \text{ }\frac{k{\sigma}^*(G)}{4n},\text{ } \frac{\sqrt{k}d_h}{2n}, \text{ } \frac{({\sigma}_k(G)-kd_h/n)^2}{2\max[k, 2d_h]}\cdot\frac{k-1}{k}\right\}.\]
If $k\ge n^{3/5}$, then ${\sigma}(C)\ge k\cdot{\sigma}^*(G)/(4n)\ge {\sigma}^*(G)/(4n^{2/5})\ge {\sigma}^*_k(G)/(4n^{2/5})$. % since ${\sigma}^*(G)$ is the maximum density of all subgraphs.
If $k\le n^{2/5}$, then ${\sigma}(C)\ge 1\ge {\sigma}^*_k(G)/k\ge {\sigma}^*_k(G)/n^{2/5}$. So we are only left with the case {of} $n^{2/5}\le k\le n^{3/5}$.

Since ${\sigma}(C)\ge  \bar\sigma/\sqrt{k}\ge {\bar\sigma}/{n^{3/10}}\ge  {\bar\sigma}/{n^{2/5}}$, we may assume $\bar\sigma<\sigma^*_k(G)/4$, and hence ${\sigma}^*_k(G)-2\bar\sigma\ge   {\sigma}^*_k(G)/2$. %{If $\bar\sigma> {\sigma}^*_k(G)/ n^{1/10} $, then ${\sigma}(C)\ge  \bar\sigma/\sqrt{k}>   {\sigma}^*_{k}(G)/(\sqrt kn^{1/10})\ge {\sigma}^*_{k}(G)/n^{2/5}$. So we may assume  $\bar\sigma\le {\sigma}^*_k(G)/ n^{1/10} $.  Hence there exists constant $c$ such that ${\sigma}^*_k(G)-2\bar\sigma\ge c\cdot {\sigma}^*_k(G)$}. %, as otherwise   ${\sigma}(C)\ge \sqrt kd_h/(2n)\ge   {\sigma}^*_k(G)/\sqrt kn^{1/10}\ge {\sigma}^*_k(G)/n^{2/5}$ and we are done.
%Hence we have that ${\sigma}^*_k(G)-kd_h/n\simeq {\sigma}^*_k(G)$, $(k-1)/k\simeq 1$ (since $k\ge n^{2/5}$) with a negligible error term.
Next we use the geometric mean to prove the performance guarantee as claimed.

In case of $k\ge 2d_h$, since  ${\sigma}^*(G)\ge {\sigma}^*_k(G)$, we have
\[{\sigma}(C)\ge\left(1\cdot\frac{k{\sigma}^*(G)}{4n}\cdot\frac{( {\sigma}^*_k(G)/2)^2}{6k}\right)^{1/3}\ge \frac{{\sigma}^*_k(G)}{5n^{2/5}},\]
In case of  $k<2d_h$, we have
\[{\sigma}(C)\ge\left(1\cdot\frac{\sqrt{k}d_h}{2n}\cdot\frac{({\sigma}^*_k(G)/2)^2}{12d_h}\cdot\frac{\sqrt{k}d_h}{2n}\cdot\frac{({\sigma}^*_k(G)/2)^2}{12d_h}\right)^{1/5}\ge \frac{{\sigma}^*_k(G)}{7 n^{2/5}},\]
where the last inequality follows from the fact that $k\ge {\sigma}^*_k(G)$.
\end{proof}

\section{Conclusion}
In Section 3, we have given four strongly polynomial time algorithms that jointly guarantee   an $O(\min\{n^{2/5},n^2/k^2\})$-approximation for the unweighted problem -- DC$k$SP. {The approximation ratio is compared with the maximum density of {\em all} $k$-subgraphs, and in this case no $O(n^{1/3-\varepsilon})$-approximation for any $\varepsilon>0$ can be {expected (recall $\Lambda\ge  n^{ 1/3}/3$} in Example \ref{eg1}(a)).} When studying  the weighted generalization -- HC$k$SP, we can extend the  techniques developed in Section~\ref{k/n}, and obtain an $O(n^2/k^2)$-approximation  for {the weighted case. Besides,   \bluecomment{the following} simple greedy approach   achieves a {$(k/2)$}-approximation}.
 \begin{algorithm}\label{alg5}
{\em Input:} connected graph $G=(V,E)$ with $|G|\ge k$ and weight $w\in \mathbb Z_+^E$. \quad
{\em Output:} a connected $k$-subgraph of $G$, denoted as {\sc Alg\ref{alg5}}($G$).
\end{algorithm}
\vspace{-0mm}\hrule
\begin{enumerate}
\vspace{-2mm}
\item \vspace{0mm} For every $v\in V$, sort  the neighbors of $v$ as $v_1, v_2,\ldots, v_t$ such that $w(vv_1)\ge w(vv_2)\ge \cdots\ge w(vv_t)$, where $t=\min\{d_G(v), k-1\}$% and  $w(vv_i)$ is the weight of edge $vv_i$;
\item \vspace{-2mm} $C^v\leftarrow G[ \{v,v_1,v_2,\ldots, v_t\}]$
 \item \vspace{-2mm} {\bf If} $|C^v|<k$, {\bf then} expand it to be a connected $k$-subgraph
 \item \vspace{-2mm} Output {\sc Alg\ref{alg5}}$(G)\leftarrow$ the heaviest  $C^v$ for all $v\in V$
\end{enumerate}
\vspace{-2mm}\hrule
\vspace{2mm}
\newpage

Notice that the weighted degree of a vertex $v$ in any heaviest $k$-subgraph of $G$ is not greater than  the weight of $C^v$ \bluecomment{constructed in Algorithm \ref{alg5}}.  {It is easy to see that Algorithm \ref{alg5}   outputs %in $O(n\log n)$ time
a connected  $k$-subgraph of $G$ whose weighted density is at least  {$2/k$} of that of the heaviest $k$-subgraph of $G$ (which is not necessarily connected). The running time is bottlenecked by the sorting at Step 1 which takes $O(|d_G(v)|\cdot\log|d_G(v)|)$ time for each $v\in V$. Hence the algorithm runs in   $O(\log n\cdot\sum_{v\in V} |d_G(v)|)=O(m\log n)$ time.}
As $\min\{n^2/k^2,k\}\le n^{2/3}$, we have   the following result.
\begin{theorem}
For any connected graph $G=(V,E)$ with weight $w\in\mathbb Z_+^E$, a connected $k$-subgraph $H$ of $G$ can be found in $O(nm)$ time such that $\sigma^*_k(G,w)/\sigma(H,w)$ $\le O(\min\{n^{2/3},n^2/k^2,k\})$, where $\sigma(H,w)$ is the weighted density of $H$, and $\sigma^*_k(G,w)$ is the weighted density of a heaviest $k$-subgraph of $G$ (which is not necessarily connected).
\end{theorem}

Since the weighted density of a graph is not necessarily related to its number of  edges or vertices, {a couple of} the results in the previous sections (such as Lemmas \ref{add},  \ref{expand} and \ref{per5}) do not hold  for the general weighted case.  Neither the techniques of extending unweighted case approximations to weighted cases in  \cite{kp93,fkp01}   apply to our setting due to the connectivity constraint. {An immediate question  is whether an $O(n^{2/5})$-approximation algorithm exists for HC$k$SP. Note from {$\Lambda_w\ge n^{ 1/2}/2$} in Example \ref{eg1}(b) that no one can achieve an $O(n^{1/2-\varepsilon})$-approximation for any $\varepsilon>0$  if she/he compares the solution value with the maximum weighted density of {\em all} $k$-subgraphs. Among other algorithmic approaches, analyzing the properties of {densest/heaviest} {\em connected} $k$-subgraphs is an important and challenging task in obtaining improved approximation ratios for DC$k$SP and HC$k$SP.} %Other future research directions for DC$k$SP and HC$k$SP include  improving approximation ratios for the general graphs, and considering special graphical topologies that are of real-world applications.

\bibliography{dense}

%\newpage

\end{document}